\documentclass[runningheads]{llncs}
\usepackage[T1]{fontenc}

\usepackage{amsmath,amssymb,stmaryrd,multicol}

\usepackage[all,cmtip]{xy}

\usepackage[shortlabels]{enumitem}

\newcommand{\program}[1]{\mathtt{#1}}
\newcommand{\At}{\mathsf{At}}
\newcommand{\LL}{\mathcal{L}}
\newcommand{\prog}{\mathsf{p}}
\newcommand{\form}{\mathsf{f}}
\newcommand{\Val}{\mathsf{Val}}
\newcommand{\otto}{\leftrightarrow}
\newcommand{\FL}{\mathsf{FL}}
\newcommand{\canon}{\mathsf{c}}
\newcommand{\pda}{\mathord{\downarrow}}
\newcommand{\can}{\mathsf{can}}
\newcommand{\GG}{\mathbf{G}}
\newcommand{\PEDAL}{\mathrm{PE}}
\newcommand{\dmff}{\mathsf{DMFF}}

\newtheorem{cor}[theorem]{Corollary}

\begin{document}
\title{Probabilistic Epistemic Dynamic Agentive Logic}

\author{Shay Allen Logan\inst{1}\orcidID{0000-0003-4066-6619}}
%
%\authorrunning{F. Author et al.}

\institute{Kansas State University, Manhattan KS 66506, USA
\email{salogan@ksu.edu}}
\maketitle              % typeset the header of the contribution
\begin{abstract}
I introduce PEDAL---a probabilistic epistemic logic meant to capture, in propositional dynamic terms, the epistemic state of an agent engaged in checking whether a program meets its specification. Semantically, PEDAL is built `on top of' PDL and uses probability measures defined on the set of possible program valuations of an otherwise-specified PDL-model. A Hilbert system with one infinitary rule is provided and proved to be sound and complete. Near the end, I discuss possible ways to circumvent infinitary proof difficulties.

\end{abstract}
\section{Introduction} 
	
	To settle some vocabulary recall that where $\At_\form$ and $\At_\prog$ are countably infinite sets of atomic formulas and atomic programs, respectively, the language, $\LL$, of propositional dynamic logic (PDL; see e.g. \cite{fischer1979propositional,harel2000dynamic}) is defined by simultaneous recursion on formulas $\gamma$ and program terms $\pi$ as follows:
	\begin{displaymath}
		\gamma := a\in\At_\form \mid \bot \mid \neg\gamma \mid \gamma\lor\gamma \mid \langle\pi\rangle\gamma\qquad
    	\pi := \program{p}\in\At_\prog \mid \pi;\pi \mid \pi\cup\pi \mid \pi^* \mid~?\gamma
	\end{displaymath}
	Other connectives are defined in the expected ways. Models for PDL are triples $\langle S,v_\form,v_\prog\rangle$ where $S$ is a nonempty set understood to represent the states of some machine, $v_\form$ is a function from $\At_\form$ to $2^S$ understood to map atomic formulas to the states where they are true, and $v_\prog$ is a function from $\At_\prog$ to $2^{S\times S}$ understood to map atomic programs to their transition relations. We let $\rho_1\circ \rho_2$ name the composition of relations $\rho_1$ and $\rho_2$ and $\rho^i$ name the $i$-fold composition of $\rho$ with itself, understood to be the identity relation when $i=0$.
	
	A `true at a state in a model' relation $\vDash$ and a function $\llbracket-\rrbracket$ mapping program terms to transition relations are then defined as follows: 

	\begin{multicols}{2}\begin{itemize}
		\item $M,s\vDash a$ iff $s\in v_\form(a)$ for $a\in\At_\form$.
		\item $M,s\not\vDash\bot$
		\item $\llbracket\program{p}\rrbracket=v_\prog(\program{p})$ for $\program{p}\in\At_\prog$.    
		\item $M,s\vDash \gamma_1\lor \gamma_2$ iff $M,s\vDash \gamma_1$ or $M,s\vDash \gamma_2$. 
		\item $M,s\vDash \neg \gamma$ iff $M,s\not\vDash \gamma$. 
		\item $\llbracket\program{\pi_1;\pi_2}\rrbracket=\llbracket\program{\pi_1}\rrbracket\circ\llbracket\program{\pi_2}\rrbracket$. 
		\item $\llbracket\program{\pi_1\cup \pi_2}\rrbracket=\llbracket\program{\pi_1}\rrbracket\cup\llbracket\program{\pi_2}\rrbracket$. 
		\item $\llbracket\program{\pi^*}\rrbracket=\cup_{i=0}^\infty\llbracket\pi\rrbracket^i$. 
		\item $\llbracket\program{?\gamma}\rrbracket=\{\langle s,s\rangle: M,s\vDash \gamma\}$. 
		\item $M,s\vDash \langle\pi\rangle\gamma$ iff for some $t$, $\langle s,t\rangle\in\llbracket\program{\pi}\rrbracket$ and $M,t\vDash \gamma$.
	\end{itemize}\end{multicols}

	%We note that it follows from these that $\program{\pi_1;\pi_2}$ represents a program that first executes $\program{\pi_1}$, then executes $\program{\pi_2}$ on the output of $\program{\pi_1}$, that $\program{\pi_1\cup \pi_2}$ represents a program that (nondeterministically) either executes $\program{\pi_1}$ or executes $\program{\pi_2}$, that $\program{\pi^*}$ represents a program that (nondeterministically) executes $\program{\pi}$ zero or more times, and that $\program{?\gamma}$ represents the test-whether-$\gamma$ program that terminates immediately in states where $\gamma$ is true and otherwise fails. Sentences of the form `$[\program{\pi}]\gamma$' and `$\langle\program{\pi}\rangle\gamma$' then represent, respectively, `$\program{\pi}$ terminates only in $\gamma$-states' and `$\program{\pi}$ terminates in at least one $\gamma$-state'.
    
    The set of \emph{valid}---which is to say always everywhere true---formulas of PDL has been axiomatized in many places, e.g. \cite{harel2000dynamic,kozen1981elementary}. Here is one such axiomatization: 
    
    \begin{multicols}{2}\begin{itemize}[leftmargin=*]%\setlength
    	\item Propositional tautologies.
		\item $[\program{\pi}](\gamma_1\to\gamma_2)\to([\program{\pi}]\gamma_1\to[\program{\pi}]\gamma_2)$
		\item $[\program{\pi}](\gamma_1\land\gamma_2)\otto([\program{\pi}]\gamma_1\land[\program{\pi}]\gamma_2)$
		\item $[\program{\pi_1\cup\pi_2}]\gamma\otto([\program{\pi_1}]\gamma\land[\program{\pi_2}]\gamma)$
		\item $[\program{?A}]B\otto(A\to B)$
		\item $[\pi_1;\pi_2]\gamma\otto[\pi_1][\pi_2]\gamma$
		\item $(\gamma\land[\pi][\pi^*]\gamma)\otto[\pi^*]\gamma$
		\item $(\gamma\land[\pi^*](\gamma\to[\pi]\gamma))\to[\pi^*]\gamma$
		\item $\gamma,\gamma\to\psi\Rightarrow\psi$
		\item $\gamma\Rightarrow[\pi]\gamma$
    \end{itemize}\end{multicols}
    Given a set of formulas in the language of PDL, we write $\Gamma\vdash\gamma$ when there is a sequence $\beta_1,\dots,\beta_n=\gamma$ such that for $1\leq i\leq n$, either $\beta_i\in\Gamma$, $\beta_i$ is an axiom, or $\beta_i$ follows from previous members of the sequence by one of the two PDL-rules. We write $\vdash\gamma$ when $\gamma$ is a theorem of PDL.
	\begin{lemma}\label{provability_facts}
		All of the following hold:
		\begin{multicols}{2}\begin{enumerate}[(i)]
			\item If $\vdash\gamma\to\bot$, then $\vdash\langle\pi\rangle\gamma\to\bot$. 
			\item If $\Gamma\cup\{\gamma_1\}\vdash\gamma_2$, then $\Gamma\vdash\gamma_1\to\gamma_2$
			\item If $\Gamma\cup\{\gamma\}\vdash\bot$, then $\Gamma\vdash\neg\gamma$
			\item $\vdash\langle\pi_1\cup\pi_2\rangle\gamma\otto(\langle\pi_1\rangle\gamma\lor\langle\pi_2\rangle\gamma)$
			\item $\vdash\langle\pi^*\rangle\gamma\otto(\gamma\lor\langle\pi\rangle\langle\pi^*\rangle\gamma)$
			\item $\vdash\langle\program{?\gamma_1}\rangle\gamma_2\otto(\gamma_1\land \gamma_2)$
			\item $\vdash\langle\pi\rangle(\gamma_1\lor\gamma_2)\otto(\langle\pi\rangle\gamma_1\lor\langle\pi\rangle\gamma_2)$
			\item $\vdash\langle\pi^n\rangle\gamma\to\langle\pi^*\rangle\gamma$ for any $n\geq 1$, where $\langle\pi^1\rangle=\langle\pi\rangle$ and $\langle\pi^{n+1}\rangle=\langle\pi;\pi^n\rangle$.
			\item $\vdash\langle\pi_1;\pi_2\rangle\gamma\otto\langle\pi_1\rangle\langle\pi_2\rangle\gamma$
		\end{enumerate}\end{multicols}
	\end{lemma} 

\section{A Motivating Example}

	Throughout we take ourselves to be discussing the beliefs of a particular agent we call Anne. To set the scene, imagine that Anne has been investigating the behavior of programs $\program{p}$ and $\program{q}$ and has come to tentatively believe the following:
	\begin{enumerate}[(i)]
		\item Whenever $A$ is true, a successful execution of $\program{p}$ will make $B$ true, and
		\item Whenever $C$ is true, a successful execution of $\program{q}$ will make $D$ true.
	\end{enumerate}
 	Suppose in fact that Anne is 60\% certain that (i) is true and 60\% certain that (ii) is true. Using the language of PDL, these are beliefs about the probability of $A\to[\program{p}]B$ being true and about the probability of $C\to[\program{q}]D$ being true. 
	
	Now let $\program{r}$ be the program $\program{(?A;p)\cup(?C;q)}$. Intuitively, $\program{r}$ acts like $\program{p}$ when $A$ is true and acts like $\program{q}$ when $C$ is true. Consider the following claim:
	\begin{enumerate}
		\item[(iii)] A successful execution of $\program{r}$ will make $B\lor D$ true.
	\end{enumerate}
	Here is a natural question that is at the heart of the work in this paper: Given that Anne is 60\% certain that (i) is true and 60\% certain that (ii) is true, how certain should she be that (iii) is true?
	
	As it turns out, absent further information Anne ought not be more than 20\% sure that (iii) is true. This is a surprising answer that will be justified below. Given that it \emph{is} surprising, though, it would be nice---in situations where judgments like these matter---to have a bit more than intuition to run on. More generally, what we might want is a rigorous, provably valid way of turning quantified credences concerning propositions like $A\to[\program{p}]B$ and $C\to[\program{q}]D$ into credences about more complex propositions like $[\program{(?A;p)\cup(?C;q)}](B\lor D)$. 
		
	%I've used the language of PDL to set up and describe this problem. But PDL itself is not helpful when it comes to answering the question for a very simple reason: PDL doesn't deal in probabilities. This a serious obstacle to applications of dynamic logics generally because \emph{proof} is usually prohibitively expensive. But many of our verification-related practices---e.g. various sorts of testing and code reviews and the like---are very nearly explicitly designed to provide us with well-informed expert judgments of the probability that a piece of code meets a particular specification. So while we \emph{have access to} what looks plausibly like probabilistic information about the behavior of programs, what we lack is a way of rigorously manipulating such information.
	
	\textbf{What I am proposing in this paper is a formal system that can be used to rigorously reason \emph{from} judgments about the probability that parts of a program meet certain specifications (expressed in PDL) \emph{to} judgments about the probability that a program meets a combined specification (also expressed in PDL).}

	The logic I provide, PEDAL, suffers from one \emph{major} drawback: it has an infinitary inference rule. While I suspect that this can be overcome, I will also note that the one infinitary rule that PEDAL uses is one that humans are \emph{often} pretty good at using. And even sans implementation, the fact that we can use PEDAL to sharpen our instincts about probabilistic reasoning with respect to program specifications seems to make further investigation of PEDAL worthwhile.
	
	%The paper is organized as follows. Section 3 gives an overview of some related work. Section 4 describes an extension of PDL with an `unlabeled' modality $\Box$, which is necessary in order to apply it in the ways outlined above. Section 5 then introduces probabilistic vocabulary and semantics. Section 6 axiomatizes the probabilistic validities and proves soundness and completeness. The axiomatization provided in this section relies on an infinitary rule. In section 7 we introduce a family of finitary restrictions of the system provided in Section 6 that, for almost all practical purposes, will suffice. The final section concludes. 

	The work done here differs \emph{philosophically} from other probabilistic treatments of dynamic logic. What PEDAL models are models \emph{of} is the state of affairs that obtains when one has some program on hand that definitely does behave in some way or other, but one is not quite sure which way. Perhaps, for example, the program is simply too complex to survey in its entirety, or to keep every detail in mind. So some facts about how it behaves are simply not present in one's mind. Or perhaps instead, one has on hand a program someone else entirely wrote, and one only has their descriptions of its behavior to go on.
	
	In situations like these, there is some fact of the matter about how the program behaves. Probabilistic reasoning is called for only because of the epistemic limitations of the agent doing the reasoning. Thus, what we are modeling here is how one ought to rationally reason from one's \emph{credences} regarding the parts of a program meeting certain specifications to \emph{credences} about the program as a whole meeting a certain specification. Situations like these are ubiquitous. Full formal verification is \emph{expensive}, so we usually settle on other means---testing of various sorts, code reviews, etc.---to fill the gap. These methods leave us \emph{more} convinced that programs behave as we want without ever \emph{fully} convincing us. We're left, in other words, with beliefs best captured by probabilities.
	
	This philosophical difference is reflected in the main formal difference between this work and other work on probabilistic dynamic logic: the probabilities here are imposed `from the outside' rather than from within a given model. More to the point, probabilistic judgments are made by placing a probability measure on (roughly) the class of all ways programs might be evaluated within a given Kripke model. This reflects the epistemic situation mentioned above: we have a program on hand that behaves in some way or other (that is, which is interpreted by some program valuation or other) but we don't know exactly which way (which is to say, we're not certain \emph{exactly which} program valuation is the correct one).
			   
\section{Reasoning About Validity}
    
	Our interest is in modeling Anne's judgments of the form `I am $x$\% certain that program $\pi$ behaves in thus and such way'. But such judgments are best modeled in PDL as judgments about whether certain formulas are \emph{valid} rather than as judgments about whether certain formulas are \emph{true}. As an example, if Anne reports that she is 75\% sure that we will end up in a state satisfying $B$ whenever we execute $\pi$ in a state satisfying $A$, then what she is reporting is that she is 75\% sure that $A\to[\pi]B$ is true \emph{everywhere} in the model that represents the system she expects to execute $\pi$ on. 
	
	So before we extend with probabilistic modals, we first need to augment PDL with an `is true everywhere' modal that we represent with an empty box `$\Box$'. The resulting system will be called PDL$^\Box$. 
%	The language of PDL$^\Box$ is, like the language of PDL, defined by simultaneous recursion on formulas $\gamma$ and programs $\pi$, this time with the following BNF:
%	\begin{displaymath}
%		\gamma := a\in\At_\form \mid \bot \mid \neg\gamma \mid \gamma\lor\gamma\mid \Box\gamma \mid \langle\pi\rangle\gamma\qquad
%    \pi := \program{p}\in\At_\prog \mid \pi;\pi \mid \pi\cup\pi \mid \pi^* \mid~?\gamma
%    \end{displaymath}
   
	What we would like from our semantics is for it to interpret everything in PDL$^\Box$ in the same way it is interpreted in PDL and for it to in addition interpret the modal $\Box$ as universal truth by saying that $M,s\vDash \Box A\text{ iff }M,t\vDash A\text{ for all }t$.
	
	But the metatheoretic work is greatly simplified by instead defining a PDL$^\Box$-model to be tuple $\langle S,v_{\form},v_{\prog},R_\Box\rangle$ with $\langle S,v_\form,v_\prog\rangle$ a PDL-model and $R_\Box$ an equivalence relation on $S$ such that for all $\program{p}\in\At_{\prog}$, $v_{\prog}(\program{p})\subseteq R_\Box$. We then use $R_\Box$ to interpret $\Box$ in the usual way: $M,s\vDash \Box A\text{ iff }M,t\vDash A\text{ for all }\langle s,t\rangle\in R_\Box$. We say that $\gamma$ is valid just if for all models $M$ and all states $s$, $M,s\vDash\gamma$.
	 	 
	Axiomatically, PDL$^\Box$ is characterized by adding the following to the axioms for PDL:	
	\begin{multicols}{2}\begin{itemize}[leftmargin=*]
    	\item $\Box(\gamma_1\to\gamma_2)\to(\Box \gamma_1\to\Box\gamma_2)$
		\item $\Box \gamma\to \gamma$
		\item $\Diamond \gamma\to\Box\Diamond \gamma$
		\item $\langle\pi\rangle \gamma\to\Diamond \gamma$
		\item $\gamma\Rightarrow\Box\gamma$
    \end{itemize}\end{multicols}
    
    %For completeness, we state aloud that `$\Diamond$' is an abbreviation of `$\neg\Box\neg$'. 
    It is a fact---a fact whose proof is entirely standard and thus omitted---that PDL$^\Box$ is sound for PDL$^\Box$-models.

	The completeness theorem is also mostly standard but is nonetheless worth going through since we will need to make reference to it in the course of the completeness proof for our target system. To begin, given a set of PDL$^\Box$-formulas $\Gamma$, define the Fischer-Ladner closure $\FL(\Gamma)$ of $\Gamma$ to be the smallest set such that
    \begin{itemize}
    	\item $\Gamma\subseteq\FL(\Gamma)$;
		\item If $\gamma_1\in\FL(\Gamma)$ and $\gamma_2$ is a subformula of $\gamma_1$, then $\gamma_2\in\FL(\Gamma)$;
		\item If $\langle\program{?\gamma_1}\rangle\gamma_2\in\FL(\Gamma)$, then $\gamma_1\in\FL(\Gamma)$;
		\item If $\langle\pi_1;\pi_2\rangle\gamma\in\FL(\Gamma)$, then $\langle\pi_1\rangle\langle\pi_2\rangle\gamma\in\FL(\Gamma)$;
		\item If $\langle\pi_1\cup\pi_2\rangle\gamma\in\FL(\Gamma)$, then $\langle\pi_1\rangle\gamma\in\FL(\Gamma)$ and $\langle\pi_2\rangle\gamma\in\FL(\Gamma)$; and
		\item If $\langle\pi^*\rangle\gamma\in\FL(\Gamma)$, then $\langle\pi\rangle\gamma\in\FL(\Gamma)$ and $\langle\pi\rangle\langle\pi^*\rangle\gamma\in\FL(\Gamma)$. 
	\end{itemize}
	We define $\FL^{\pm}(\Gamma)=\FL(\Gamma)\cup\{\neg\gamma:\gamma\in\FL(\Gamma)\}$ and let $S(\Gamma)$ be the set of maximal consistent subsets of $\FL^{\pm}(\Gamma)$
	%---so if $W\in S(\Gamma)$, then $W\subseteq\FL^{\pm}(\Gamma)$, $W\not\vdash\bot$, and if $W\subseteq W'\subseteq\FL^{\pm}(\Gamma)$, then either $W'=W$ or $W'\vdash\bot$. 
	In the remainder, we will abuse notation by failing to distinguish, for $W\in S(\Gamma)$, between the \emph{set} $W$ and \emph{sentence} $\bigwedge W$, both of which will be referred to simply as $W$.\medskip
	
	\noindent\textbf{Fact:} If $\Gamma$ is finite, then so is $\FL(\Gamma)$ and thus $S(\Gamma)$.\smallskip
    
    Proofs of the following lemmas are straightforward and thus omitted:
    \begin{lemma}\label{equiv_lemma}
		For all $\gamma\in\FL^\pm(\Gamma)$, if $\gamma\not\vdash\bot$, then $\vdash\gamma\otto\bigvee_{\gamma\in W\in S(\Gamma)}\bigwedge W$
	\end{lemma}

	\begin{lemma}\label{connectiveswork}
		For all $W\in S(\Gamma)$, 
		\begin{enumerate}
			\item If $\neg\gamma\in\FL^{\pm}(\Gamma)$ then $\gamma\not\in W$ iff $\neg\gamma\in W$;
			\item If $\gamma_1\lor\gamma_2\in\FL^{\pm}(\Gamma)$ then $\gamma_1\lor\gamma_2\in W$ iff $\gamma_1\in W$ or $\gamma_2\in W$. 
		\end{enumerate}
	\end{lemma}
	
	Since `$S(\Gamma)$' can almost always be inferred from context, we often omit it.
    
    \begin{lemma}\label{in_out_lemma}
    	If $\mathbf{G}\subseteq S(\Gamma)$, then $\vdash(\bigwedge_{W\in\mathbf{G}}\neg W)\otto(\bigvee_{V\not\in\mathbf{G}}V)$.
    \end{lemma}
    \begin{proof}
		For `$\to$', first note that by Lemma~\ref{equiv_lemma}, $\vdash\gamma\otto\bigvee_{\gamma\in W}W$ and by Lemma~\ref{equiv_lemma} and Lemma~\ref{connectiveswork}, $\vdash\neg\gamma\otto\bigvee_{\gamma\not\in W}W$. It follows that $\vdash \top \otto \bigvee_{W\in S(\Gamma)} W$.  Thus $\vdash\top\to(\bigvee_{W\in\GG} W\lor\bigvee_{V\not\in\GG} V)$. So $\vdash\neg\bigwedge_{W\in\GG}\neg W\lor\bigvee_{V\not\in\GG} V$, which is to say that $\vdash\bigwedge_{W\in\GG}\neg W\to\bigvee_{V\not\in\GG} V$. 
		
		For the other direction, observe that if $V\not\in\GG$ but $W\in\GG$, then by maximality, $\vdash (V\land W)\to\bot$. Since this is true for all $W\in\GG$, it follows that $\vdash (V\land\bigvee_{W\in\GG} W)\to\bot$. And since this is true for all $V\not\in\GG$, it follows that $\vdash(\bigvee_{V\not\in\GG} V\land\bigvee_{W\in\GG} W)\to\bot$. It follows that $\vdash\neg(\bigvee_{V\not\in\GG} V\land\neg\bigwedge_{W\in\GG} \neg W)$. So $\vdash\bigvee_{V\not\in\GG} V\to\bigwedge_{W\in\GG} \neg W$.
    \end{proof}
    
    Completeness proofs for ordinary PDL use $S(\Gamma)$ (where $\Gamma$ is a consistent set of formulas) to construct a model $\langle S(\Gamma),v^\canon_\form,v^\canon_\prog\rangle$ where $v^\canon_\form(p)=\{W\in S(\Gamma) : p\in W\}$ and $\langle W_1,W_2\rangle\in v^\canon_\prog(\program{p})$ iff $W_1\land\langle\program{p}\rangle W_2\not\vdash\bot$. For our purposes, we instead use the model $M^\canon(\Gamma)=\langle S(\Gamma),v^\canon_\form,v^\canon_\prog,R^\canon_\Box\rangle$ where the first three elements are exactly as in the ordinary PDL case and $\langle W_1,W_2\rangle\in R^\canon_\Box$ iff $W_1\land\Diamond W_2\not\vdash\bot$.
    
    \begin{lemma}\label{extension}
    	If $W\in S(\Gamma)$ and $W'\in S(\Gamma)$, then for all programs $\pi$, if $W\land\langle\pi\rangle W'\not\vdash\bot$, then $\langle W,W'\rangle\in \llbracket\pi\rrbracket^\canon$.
    \end{lemma}
    \begin{proof}
    	By induction on $\pi$. The base case follows from the definition of $v^\canon_\prog$ and the test and choice cases are straightforward. 
	
		%Suppose $\pi=\program{?\gamma}$ and $W\land\langle\program{?\gamma}\rangle W'\not\vdash\bot$. Then by (vi) of Lemma~\ref{provability_facts}, $W\land(\gamma\land W')\not\vdash\bot$. So (i) $W\land W'\not\vdash\bot$ and (ii) $W\land \gamma\not\vdash\bot$. From (i) and the fact that both $W$ and $W'$ are maximal consistent subsets of $\FL^{\pm}(\Gamma)$, we must have that $W=W'$. From (ii) and the fact that $W$ is maximal consistent we have that $\gamma\in W$. So $\langle W,W'\rangle\in\llbracket\pi\rrbracket^\canon$.
		
		%Suppose $\pi=\pi_1\cup\pi_2$ and $W\land\langle\pi_1\cup\pi_2\rangle W'\not\vdash\bot$. Then by (iv) of Lemma~\ref{provability_facts} and some distribution, we have that $(W\land\langle\pi_1\rangle W')\lor(W\land\langle\pi_2\rangle W')\not\vdash\bot$. It follows that either $(W\land\langle\pi_1\rangle W')\not\vdash\bot$ or $(W\land\langle\pi_2\rangle W')\not\vdash\bot$. So by the inductive hypothesis, either $\langle W,W'\rangle\in\llbracket\pi_1\rrbracket$ or $\langle W,W'\rangle\in\llbracket\pi_2\rrbracket$. Either way, $\langle W,W'\rangle\in\llbracket\pi_1\rrbracket\cup\llbracket\pi_2\rrbracket=\llbracket\pi_1\cup\pi_2\rrbracket$. 
		
		Suppose $\pi=\pi_1;\pi_2$ and $W\land\langle\program{\pi_1;\pi_2}\rangle W'\not\vdash\bot$. Then $W\land\langle\pi_1\rangle(\langle\pi_2\rangle W'\land\langle\pi_2\rangle W')\not\vdash\bot$. By Lemma~\ref{equiv_lemma}, then, we have that $W\land\langle\pi_1\rangle\left(\left(\bigvee_{\langle\pi_2\rangle W'\in W''}W''\right)\land\langle\pi_2\rangle W'\right)\not\vdash\bot$. After a bit of rearranging this gives that $\bigvee_{\langle\pi_2\rangle W'\in W''}(W\land\langle\pi_1\rangle(W''\land\langle\pi_2\rangle W'))\not\vdash\bot$. Thus, for at least one $W''$ with $\langle\pi_2\rangle W'\in W''$, $W\land\langle\pi_1\rangle(W''\land\langle\pi_2\rangle W')\not\vdash\bot$. It follows that $W\land\langle\pi_1\rangle W''\not\vdash\bot$ and, after an application of (i) of Lemma~\ref{provability_facts}, that $W''\land\langle\pi_2\rangle W'\not\vdash\bot$. So by the inductive hypothesis, $\langle W,W''\rangle\in\llbracket\pi_1\rrbracket$ and $\langle W'',W'\rangle\in\llbracket\pi_2\rrbracket$. So $\langle W,W'\rangle\in\llbracket\pi_1;\pi_2\rrbracket$.
		
		Suppose $\pi=\pi^*$ and $W\land\langle\pi^*\rangle W'\not\vdash\bot$. Let $\mathbf{G}$ be the smallest subset of $S(\Gamma)$ such that (a) $W\in \GG$ and (b) if $V\in \GG$ and $V\land\langle\pi\rangle V'\not\vdash\bot$, then $V'\in \GG$. Observe that if $U\in \GG$, then by repeated application of the inductive hypothesis, $\langle W,U\rangle\in\llbracket\pi^*\rrbracket$. Now let $Y=\bigvee_{g\in\GG} g$. By Lemma~\ref{in_out_lemma}, $\neg Y$ is equivalent to $\bigvee_{h\not\in\GG} h$. So $Y\land\langle\pi\rangle\neg Y$ is equivalent to $\bigvee_{\substack{g\in \GG\\h\not\in \GG}}g\land\langle\pi\rangle h$. It follows from the definition of $\GG$ that each disjunct is inconsistent. Thus the disjunction as a whole is inconsistent. So $\neg(Y\land\langle\pi\rangle\neg Y)$ is a theorem. But then $Y\to[\pi]Y$ is also a theorem and thus so is $[\pi^*](Y\to[\pi]Y)$. So $W\to[\pi^*](Y\to[\pi]Y)$ is a theorem as well. Note also that since $W\in\GG$ by construction and $Y$ is the disjunction of the members of $\GG$, $W\to Y$ is also a theorem. It follows that $W\to(Y\land[\pi^*](Y\to[\pi]Y))$ is a theorem. So by the induction axiom, $W\to[\pi^*]Y$ is a theorem. It follows that $W\land\neg[\pi^*]Y\vdash\bot$, and thus that $W\land\langle\pi^*\rangle\neg Y\vdash\bot$. So $W\land\langle\pi^*\rangle\bigvee_{h\not\in\GG}h\vdash\bot$. 
		
		After a by-now-familiar collection of rearrangements, $\bigvee_{h\not\in\GG}(W\land\langle\pi^*\rangle h)\vdash\bot$. Thus, for each $h\not\in\GG$, $W\land\langle\pi^*\rangle h\vdash\bot$. But then since $W\land\langle\pi^*\rangle W'\not\vdash\bot$, we must have that $W'\in\GG$. So $\langle W,W'\rangle\in\llbracket\pi^*\rrbracket$ as required. 
    \end{proof}
    
    \begin{lemma}\label{programswork}
    	For all $W\in S(\Gamma)$, if $\langle\program\pi\rangle\gamma\in\FL^{\pm}(\Gamma)$ then $\langle\pi\rangle\gamma\in W$ iff for some $W'\in S(\Gamma)$, $\langle W,W'\rangle\in \llbracket\pi\rrbracket^\canon$ and $\gamma\in W'$.
    \end{lemma}
    \begin{proof}
    	Left-to-right is straightforward.
		%For left-to-right, suppose $\langle\pi\rangle\gamma\in W$. Then $W\land\langle\pi\rangle\gamma\not\vdash\bot$. Then by Lemma~\ref{equiv_lemma}, $W\land\langle\pi\rangle\bigvee_{\gamma\in W'}W'\not\vdash\bot$. So $\bigvee_{\gamma\in W'}(W\land\langle\pi\rangle W')\not\vdash\bot$. So for at least one $W'$ with $\gamma\in W'$, $W\land\langle\pi\rangle W'\not\vdash\bot$ which, by the previous lemma, gives that $\langle W,W'\rangle\in\llbracket\pi\rrbracket^\canon$. 
		For right-to-left, we do an induction on $\pi$. For the base case, suppose for some $W'\in S(\Gamma)$ that $\langle W,W'\rangle\in v^\canon_\prog(\program{p})$ and $\gamma\in W'$. Then $W\land\langle\program{p}\rangle W'\not\vdash\bot$. It follows that $W\land\langle\program{p}\rangle\gamma\not\vdash\bot$. Thus since $W$ is maximal, $\langle\program{p}\rangle\gamma\in W$. For the induction, only the star case is worth examining
				
		%Suppose that $\langle W,W'\rangle\in\llbracket\pi_1;\pi_2\rrbracket$ and $\gamma\in W'$. Then there is $W''$ so that (i) $\langle W,W''\rangle\in\llbracket\pi_1\rrbracket$ and (ii) $\langle W'',W'\rangle\in\llbracket\pi_2\rrbracket$. By the inductive hypothesis, (ii) gives that $\langle\pi_2\rangle\gamma\in W''$. This, together with (i) and the inductive hypothesis then gives that $\langle\pi_1\rangle\langle\pi_2\rangle\gamma\in W$, from which it follows that $\langle\pi_1;\pi_2\rangle\gamma\in W$ as required.
		
		Suppose that $\langle W,W'\rangle\in\llbracket\pi^*\rrbracket$ and $\gamma\in W'$. Then either $W=W'$---in which case (v) of Lemma~\ref{provability_facts} gives that $\langle\pi^*\rangle\gamma\in W$ as required---or there are $V_1,\dots,V_n=W'$ so that $\langle W,V_1\rangle\in\llbracket\pi\rrbracket$ and for $1\leq i\leq n$, $\langle V_i,V_{i+1}\rangle\in\llbracket\pi\rrbracket$. We show by induction on $n$ that in all such cases, $W\land\langle\pi^n\rangle\gamma\not\vdash\bot$. We distinguish this \emph{inner induction} that applies only to this from what we will now call the \emph{outer induction} being used to prove the lemma as a whole.
		
		If $n=1$, then $V_1=W'$ and $\langle W,W'\rangle\in\llbracket\pi\rrbracket$. So by the inductive hypothesis, $\langle\pi\rangle\gamma\in W$. Thus since $W$ is consistent, $W\land\langle\pi^1\rangle\gamma\not\vdash\bot$. Suppose $n=k+1$. Since $\langle V_k,W'\rangle\in\llbracket\pi\rrbracket$, the outer inductive hypothesis gives that $\langle\pi\rangle\gamma\in V_k$. By the inner inductive hypothesis, then, $W\land\langle\pi^k\rangle\langle\pi\rangle\gamma\not\vdash\bot$. But then by (ix) of Lemma~\ref{provability_facts}, $W\land\langle\pi^{k+1}\rangle\gamma\not\vdash\bot$. So by the inner induction, $W\land\langle\pi^n\rangle\gamma\not\vdash\bot$ for all $n$. Thus by (viii) of Lemma~\ref{provability_facts}, $W\land\langle\pi^*\rangle\gamma\not\vdash\bot$. So since $W$ is maximal, $\langle\pi^*\rangle\gamma\in W$, concluding the outer induction.
    \end{proof}
    
    \begin{lemma}\label{lemma_box_works}
    	For all $W\in S(\Gamma)$, if $\Box\gamma\in\FL^{\pm}(\Gamma)$ and $\Box\gamma\not\vdash\bot$, then $\Box\gamma\in W$ iff for all $W'\in S(\Gamma)$, if $\langle W,W'\rangle\in R^\canon_\Box$ then $\gamma\in W'$.
	\end{lemma}
	\begin{proof}
		Suppose $\Box\gamma\in W$, $W'\in S(\Gamma)$, and $\gamma\not\in W'$. Then $\neg\gamma\in W'$. So $\Diamond W'\vdash\neg\Box\gamma$. Thus $W\land\Diamond W'\vdash\bot$. So $\langle W,W'\rangle\not\in R^\canon_\Box$. 
		
		Now suppose that $\Box\gamma\not\in W$. Then $\neg\Box\gamma\in W$. So $\Diamond\neg\gamma\in W$. Thus since $W$ is consistent, $W\land\Diamond\neg\gamma\not\vdash\bot$. But by Lemma~\ref{equiv_lemma}, we then have that $W\land\Diamond\bigvee_{\neg\gamma\in W'}W'\not\vdash\bot$. So for at least one $W'$ with $\neg\gamma\in W'$ we must have that $W\land\Diamond W'\not\vdash\bot$, and thus $\langle W,W'\rangle\in R^\canon_\Box$. And since $\neg\gamma\in W'$, $\gamma\not\in W'$. So it's not the case that for all $W'$, if $\langle W,W'\rangle\in R^\canon_\Box$, then $\gamma\in W'$.
	\end{proof}

	Reminder: the canonical model $M^\canon(\Gamma)$ is the tuple $\langle S(\Gamma),v^\canon_\form,v^\canon_\prog,R^\canon_\Box\rangle$ where $S(\Gamma)$, $v^\canon_\prog$, and $R^\canon_\Box$ are as we've been discussing and $v^\canon_{\form}(a)=\{W\in S(\Gamma):a\in W\}$. Since this construction precisely mimics, in the first three components, the construction of the usual canonical model for $\Box$-free PDL, it's clear that $M^\canon(\Gamma)$ is a PDL-model. For similar reasons, it's clear that $M^\canon(\Gamma)$ is an S5-model

	\begin{lemma}\label{its_a_model}
		$M^\canon(\Gamma)$ is a PDL$^\Box$-model.
	\end{lemma}
	\begin{proof}
		Given the comments in the preceding paragraph, it suffices to ensure that  for all $\program{p}\in\At_{\prog}$, $v_{\prog}(\program{p})\subseteq R_\Box$. To that end, let $\langle W_1,W_2\rangle\in v^\canon_\prog(\program{p})$. Then $W_1\land\langle\program{p}\rangle W_2\not\vdash\bot$. But then since $\vdash\langle\program{p}\rangle W_2\to\Diamond W_2$, we also have that $W_1\land\Diamond W_2\not\vdash\bot$. So $\langle W_1,W_2\rangle\in R^\canon_\Box$, and thus $v^\canon_\prog(\program{p})\subseteq R^\canon_\Box$. 
	\end{proof}
	
	\begin{lemma}\label{pdl_box_key_lemma}
		If $\gamma\in\FL^{\pm}(\Gamma)$, then $M^\canon(\Gamma),W\vDash\gamma$ iff $\gamma\in W$. 
	\end{lemma}
	\begin{proof}
		By induction on $\gamma$. The base case is immediate and the induction steps for $\neg$ and $\lor$ are straightforward. The induction step for $\langle\pi\rangle$ is an immediate consequence of Lemma~\ref{programswork} and the inductive hypothesis. The induction step for $\Box$ is an immediate consequence of Lemma~\ref{lemma_box_works} and the inductive hypothesis.
	\end{proof}

	\begin{theorem}
		PDL$^\Box$ is complete for PDL$^\Box$-models.
	\end{theorem}
	\begin{proof}
		Suppose that $\not\vdash \gamma$. Then $\neg \gamma$ is consistent. So there is a world $W$ in $M^\canon(\{\gamma\})$ such that $M^\canon(\{\gamma\}),W\vDash\neg \gamma$. So $M^\canon(\{\gamma\}),W\not\vDash \gamma$. Thus $\gamma$ is not valid.
	\end{proof}
	
\section{Adding Probabilities Part 1: Semantics}

	In the remainder, we will refer to formulas in the language of PDL$^\Box$ as \emph{ground formulas}. The language we now turn to contains \emph{terms} of the form `$\Pr(\gamma)$' where $\gamma$ is a ground formula and `$\Pr$' is meant to be a `probability of' operator. The set of atomic formulas, $\At_1$, is the set of all expressions of the form $\Pr(\gamma)\geq q$ where $\gamma$ is a ground formula and $0\leq q\leq 1$ is a rational number. The language $\LL_1$ is set of the formulas $\gamma_1$ given in BNF as $\gamma_1 := \At_1 \mid \neg\gamma \mid \gamma\lor\gamma$

	We adopt abbreviations for $\land$, $\to$, and $\otto$ as expected, as well as of the following. It can be (but won't be) shown that all of these behave as expected.
	\setlength{\columnsep}{0pt}\begin{multicols}{2}
		\begin{itemize}[leftmargin=*]
			\item $q\leq\Pr(\gamma)~:=~\Pr(\gamma)\geq q$
			\item $\Pr(\gamma)<q~:=~\neg\Pr(\gamma)\geq q$
			\item $q>\Pr(\gamma)~:=~\Pr(\gamma)<q$
			\item $q\geq\Pr(\gamma)~:=~\Pr(\neg\gamma)\geq 1-q$
			\item $\Pr(\gamma)\leq q~:=~q\geq\Pr(\gamma)$
			\item $q=\Pr(\gamma)~:=~\Pr(\gamma)=q$
			\item $\Pr(\gamma)=q~:=~\Pr(\gamma)\geq q\land q\geq\Pr(\gamma)$
			\item $\Pr(\gamma)> q~:=~\Pr(\gamma)\geq q\land\neg\Pr(\gamma)=q$
			\item $q<\Pr(\gamma)~:=~\Pr(\gamma)>q$
		\end{itemize}
	\end{multicols}    
	\noindent Note that $\LL_1$ does not allow dynamic modals to occur outside probability operators, or iterated probability operators. Probabilistic logics based on languages like $\LL_1$ are investigated in Chapter 3 of \cite{ognjanovic2016probability} among many other places.

	\subsection{Models}

		A PEDAL model is a five-tuple $\langle S, v_{\form}, R_\Box, \Pi,\mu\rangle$ where the initial triple is an S5-model, $\Pi$ is a partition of $\Val^S_{\prog}$---the set of $R_\Box$-compatible functions from $\At_\prog$ to $2^{S\times S}$---into disjoint sets and $\mu$ is a real-valued function from $\Pi$ to $[0,1]$ for which $\sum_{P\in\Pi}\mu(P)=1$. 

		The tail pair $\langle\Pi,\mu\rangle$ represents the epistemic state of our agent Anne. To see how, it helps to imagine Anne's investigations of $\program{p}$ have left her with beliefs of the form `I am 80\% sure that $\program{p}$ instantiates transition relation $R$, but think there is a 20\% chance that it instead instantiates transition relation $S$'. Being sensible, Anne also knows that $\program{q}$ does instantiate some transition relation or other. So we should take her claim `I am 80\% sure that $\program{p}$ instantiates transition relation $R$' as elliptical for a much longer claim of the form ``I am 80\% sure that (either $\program{p}$ instantiates $R$ and $\program{q}$ instantiates $T_1$ or $\program{p}$ instantiates $R$ and $\program{q}$ instantiates $T_2$ or \dots or $\program{p}$ instantiates $R$ and program $\program{q}$ instantiates $T_n$.)''

		Of course, the elliptical we \emph{actually} take Anne to be endorsing says similar things about all the other program variables as well. Letting $V$ be the set of program valuations that map $\program{p}$ to $R$, then we can capture this by saying that Anne's claim to be 80\% sure that program $\program{p}$ instantiates transition relation $R$ is actually elliptical for the claim that she is 80\% sure that the true program valuation is a member of $V$.
		
		Now suppose (implausibly) that Anne's epistemic state is \emph{exhausted} by her thinking it 80\% likely that $\program{p}$ instantiates $R$ and 20\% likely that it instead instantiates $S$. Then we could model Anne's epistemic state by (a) the three-element partition $\{\{v: v(p)=R\},\{v:v(p)=S\},\{\text{everything else}\}\}$ and (b) the function that maps the first set to 0.8, the second to 0.2, and the third to 0.
		
		Of course, in general Anne's epistemic state will be much more complicated. But while this complication will make the model more complex, it won't change the details of how the model works---it will remain the case that the cells in $\Pi$ represent sets of valuations Anne is indifferent between, and that the function $\mu$ assigns to each cell a probability that represents Anne's degree of belief in the claim that the true program valuation is found in that cell. 
		
		Throughout we will enforce the following assumptions: (1) $\Pi$ is finite---that is, partitions $\Val^S_{\prog}$ into only finitely many cells; (2) $\mu$ is supported only on finite cells---that is, for each $P\in\Pi$, $\mu(P)\neq 0$ only if $P$ is finite. The first of these reflects the fact that Anne is finite and, as a result, doesn't have beliefs about infinitely many different things at once. The second reflects the fact that the space of possible transition relations on machines she can practically build---or, if you prefer, which can fit in the universe---is also finite.

	\subsection{Semantics}

		The measure $\mu$ is defined on sets of valuations. The formulas in $\LL_1$ attribute probabilities to formulas. So in order to provide semantics for $\LL_1$, we first need to somehow or other understand $\mu$ as assigning probabilities to formulas. 
		
		It's clear enough what we \emph{want} to have happen here: let $M=\langle S,v_{\form},R_\Box,\Pi,\mu\rangle$ be a PEDAL-model, $\gamma$ be a ground formula, and $s\in S$. For each $v\in\Val^S_\prog$, let $M(v)$ be the PDL$^\Box$-model $\langle S,v_{\form},R_\Box,v\rangle$ and write $\Val^S_\prog(s,\gamma)$ for $\{v\in\Val^S_\prog:M(v),s\vDash\gamma\}$. Then what we \emph{want} is for $\mu(s,\gamma)$ to be $\mu(\Val^S_\prog(s,\gamma))$. The problem is that this doesn't quite make sense---we have no good reason to expect $\Val^S_\prog(s,\gamma)$ to be a member of $\Pi$. And since $\mu$ isn't defined anywhere else, $\mu$ will, in general, be undefined when applied to $\Val^S_\prog(s,\gamma)$.

		The solution to this problem is to exploit Anne's indifference. Recall that the individual valuations in any particular cell in a partition are valuations that we understand Anne to be (epistemically) indifferent between. So even though Anne has \emph{explicitly} assigned probabilities only to \emph{sets of valuations}, it is plausible to take her indifference to mean that she has \emph{implicitly} assigned a probability to each \emph{particular valuation}. Concretely, given a cell $P\in\Pi$ and a valuation $v\in P$, we will take the probability, $\mu(\{v\})$, that Anne assigns to the particular valuation $v$ to be $\frac{\mu(P)}{|P|}$. Note that our finiteness conditions guarantee that this is always well-defined. Note also that this gives us a well-defined assignment of probabilities to \emph{every} subset $V$ of $\Val^S_\prog$ by saying that $\mu(V):=\sum_{\mu(P)\neq 0}\left(\mu(P)\cdot\frac{|P\cap V|}{|P|}\right)$. Finally, note that since $\Pi$ partitions $\Val^S_\prog$, this sum can be rewritten as $\sum_{v\in V}\mu(\{v\})$. From this latter representation, the following lemma (and its corollary) are essentially immediate:

		\begin{lemma}
			If $X$ and $Y$ are finite subsets of $\Val^S_\prog$ and $X\cap Y=\emptyset$, then $\mu(X\cup Y)=\mu(X)+\mu(Y)$.
		\end{lemma} 

		\begin{cor}\label{oneminus}
			If $M$ is a PEDAL-model, $s$ is an $M$-state, and $\gamma$ is a ground formula, then $\mu(s,\gamma)+\mu(s,\neg\gamma)=1$	
		\end{cor}

		Truth in a PEDAL model is defined as follows:
		\begin{multicols}{2}\begin{itemize}
			\item $M,s\vDash \Pr(\gamma) \geq q$ iff $\mu(s,\gamma) \geq q$.
			\item $M,s\vDash \gamma\lor\psi$ iff $M,s\vDash\gamma$ or $M,s\vDash\psi$.
			\item $M,s\vDash\neg\gamma$ iff $M,s\not\vDash\gamma$.
		\end{itemize}\end{multicols}
		\noindent Say that $A$ is valid when $M,s\vDash A$ for all PEDAL-models $M$ and all $M$-states $s$.

			\begin{lemma}\label{oddity}
				If $\gamma$ is a dynamic-modal free formula (dmff), $M=\langle S,v_\form,R_\Box,\Pi,\mu\rangle$ is a PEDAL-model, $s\in S$, and $\{v,w\}\subseteq\Val^S_\prog$, then $M(v),s\vDash\gamma$ iff $M(w),s\vDash\gamma$
			\end{lemma}
			\begin{proof}
				By induction on $\gamma$.
			\end{proof}

			\begin{cor}
				If $\gamma$ is a dmff, then either $\mu(s,\gamma)=1$ or $\mu(s,\gamma)=0$.
			\end{cor}
%			\begin{proof}
%				Let $\gamma$ be any dmff. Then $\mu(s,\gamma)=\mu(\Val^S_\prog(s,\gamma))$. But since $\gamma$ is free of dynamic modals, $M(v),s\vDash\gamma$ iff $M(w),s\vDash\gamma$ for all $v$ and $w$. It follows that either $\Val^S_\prog(s,\gamma)=\Val^S_\prog$ or $\Val^S_\prog(s,\gamma)=\emptyset$. 
%			\end{proof}

\section{Adding Probabilities Part 2: An Axiomatization}

	We axiomatize PEDAL as follows, where $\gamma$ ranges over ground formulas and $r_1$ and $r_2$ over $[0,1]\cap\mathbb{Q}$.
	\begin{enumerate}[({A}1), leftmargin=*]
		\item All $\LL_1$-instances of propositional tautologies.
		\item $\Pr(\gamma)\geq 0$.
		\item $\Pr(\gamma)=1\lor\Pr(\gamma)=0$ if $\gamma$ is a dmff.
		\item $\Pr(\gamma)\leq r_1\to\Pr(\gamma)<r_2$ when $r_1<r_2$.
		\item $\Pr(\gamma)<r\to\Pr(\gamma)\leq r$
		\item $(\Pr(\gamma_1)\geq r_1\land\Pr(\gamma_2)\geq r_2\land\Pr(\gamma_1\land\gamma_2)=0)\to\Pr(\gamma_1\lor\gamma_2)\geq \min(1,r_1+r_2)$
		\item $(\Pr(\gamma_1)\leq r_1\land\Pr(\gamma_2)<r_2)\to\Pr(\gamma_1\lor\gamma_2)<r_1+r_2$ when $r_1+r_2\leq 1$
		\item[(R1)] $A$ and $A\to B$ infer $B$.
		\item[(R2)] From $\vdash_{\text{PDL}^\Box}\gamma$ infer $\Pr(\gamma)= 1$
		\item[(R3)] For $r>0$, from $A\to\Pr(\gamma)\geq r-\frac{1}{k}$ for all $k\geq\frac{1}{r}$, infer $A\to\Pr(\gamma)\geq r$.
	\end{enumerate}
	Note that R3 is an infinitary rule. We say that $X\vdash A$ just if there is a sequence $B_1,\dots, B_{\lambda+1}=A$ (where $\lambda$ is \emph{a finite or countable ordinal}) such that for $1\leq i\leq\lambda+1$, either (i) $B_i$ is an axiom or (ii) there are $j<i$ and $k<i$ so that $B_j$ is $B_k\to B_i$ or (iii) for some $\gamma$, $\vdash_{\text{PDL}^\Box}\gamma$ and $B_i$ is $\Pr(\gamma)=1$ or (iv) for all $k\geq\frac{1}{r}$ there is $j_k<i$ so that $B_k$ is $A\to\Pr(\gamma)\geq r-\frac{1}{k}$ and $B_i$ is $A\to\Pr(\gamma)\geq r$. We say $\vdash A$ just when $\emptyset\vdash A$.
	
	\begin{theorem}
		PEDAL is sound for PEDAL-models.
	\end{theorem}
	\begin{proof}
		By straightforward induction on derivations. %Of the axioms, note that A3 follows from (the discussion following) Lemma~\ref{oddity}. Checking the other axioms is entirely straightforward; for the sake of demonstrating this, we show how to check A6.\footnote{Note that in the statement of (A6), the `$\min(1,r_1+r_2)$ is present for well-formedness reasons. In any instance of (A6) with a true antecedent, it isn't needed. But there are instances of (A6) with false antecedents where the absence of the $\min$ would leave us with non-formulas; e.g.
%			$(\Pr(\gamma_1)\geq 1\land\Pr(\gamma_2)\geq 1\land\Pr(\gamma_1\land\gamma_2)=0)\to\Pr(\gamma_1\lor\gamma_2)\geq \min(1,2)$.}
%		
%		Suppose $\mu(s,\gamma_1)\geq r_1$, and $\mu(s,\gamma_2)\geq r_2$ and that $\mu(s,\gamma_1\land\gamma_2)=0$. Recall that for finite $P$, $|P\cap(Q\cup R)|=|P\cap Q|+|P\cap R|-|P\cap(Q\cap R)|$. We compute $\mu(s,\gamma_1\lor\gamma_2)$ as follows:
%		{\small\begin{align*}
%		\mu(s,\gamma_1\lor\gamma_2)
%			& = \sum_{\mu(P)\neq 0}\mu(P)\frac{|P\cap\Val^S_\prog(s,\gamma_1\lor\gamma_2)|}{|P|} \\
%			& = \sum_{\mu(P)\neq 0}\mu(P)\frac{|P\cap(\Val^S_\prog(s,\gamma_1)\cup\Val^S_\prog(s,\gamma_2))|}{|P|} \\
%			& = \sum_{\mu(P)\neq 0}\mu(P)\frac{|P\cap \Val^S_\prog(s,\gamma_1)|+|P\cap\Val^S_\prog(s,\gamma_2)|-|P\cap\Val^S_\prog(s,\gamma_1)\cap\Val^S_\prog(s,\gamma_2)|}{|P|} \\
%			& = \mu(s,\gamma_1)+\mu(s,\gamma_2)-\mu(s,\gamma_1\land\gamma_2) \\
%			& = \mu(s,\gamma_1)+\mu(s,\gamma_2)\\
%			& \geq r_1+r_2.			
%		\end{align*}}
%		
%		
%		That rule R1 preserves validity is clear. For R2, suppose $\vdash_{\text{PDL}^\Box}\gamma$. Then $\Val^S_\prog(s,\gamma)=\Val^s_\prog$. So $\mu(s,\gamma)=1$ as required. For R3, let $r>0$ and suppose that $M,s\not\vDash A\to\Pr(\gamma)\geq r$. Then $M,s\vDash A$ but $\mu(s,\gamma)<r$. Choose $k\geq\frac{1}{r}$ so that $\frac{1}{k}<r-\mu(s,\gamma)$. Then $\mu(s,\gamma)<r-\frac{1}{k}$, so $M,s\not\vDash A\to\Pr(\gamma)\geq r-\frac{1}{k}$. 
	\end{proof}
	
	\subsection{Completeness}
	
	We defined $\FL(X)$ above only with respect to $X\subseteq\LL$. If $X$ is a subset of $\LL_1$ rather than $\LL$ then we define $\FL(X)$ to be $\FL(\{\gamma:\Pr(\gamma)\geq q$ occurs in $A$ for some $A\in X\})$. We can then define $\FL^{\pm}(X)$ and $S(X)$ in the same way we did for PDL$^\Box$. Note that it remains the case that if $X$ is finite, then so is $\FL(X)$.
	
	Say that two members $W_1$ and $W_2$ of $S(X)$ are dmff-compatible when they contain exactly the same dmffs. Given $W\in S(X)$, write $\dmff(W)$ for the set of all members of $S(X)$ that are dmff-compatible with $W$. A \emph{section} of $S(X)$ is a function $\sigma:S(X)\to S(X)$ mapping each $W\in S(X)$ to a member of $\dmff(W)$. Each section $\sigma$ gives rise to a program valuation $v_\sigma\in\Val^{S(X)}_\prog$ defined by $v_\sigma(\program{p})=\{\langle W_1,W_2\rangle: \sigma(W_1)\land\langle\program{p}\rangle \sigma(W_2)\not\vdash\bot\}$. We write $\llbracket\pi\rrbracket_\sigma$ for the usual recursive extension of $v_\sigma$ to arbitrary programs. 
	
	Some intuition is helpful at this point. Let $M$ be a PDL$^\Box$- (not PEDAL-) model. Each of $M$'s states can be thought of as containing two types of information: non-dynamic information and dynamic information. Distinct states in $\dmff(W)$ contain the same non-dynamic information, but disagree about dynamic information. The fact that \emph{there are} these different possibilities is sometimes crucial to the canonical model behaving as it does. This leaves us with something of a dilemma when it comes to our completeness proof. On the one hand, for our canonical model to behave correctly, we will, typically, need to keep all of these distinct-but-non-dynamically-indistinguishable points around. On the other hand, in PEDAL-models, dynamic behavior of points isn't \emph{fixed} by a model. Instead, PEDAL models explicitly leave all sorts of options open by not choosing a program valuation. Concretely, we say that \emph{with respect to the section $\sigma$}, execution of $\program{p}$ can take us from $W_1$ to $W_2$ not when $W_1\land\langle\program{p}\rangle W_2\not\vdash\bot$ (as we did in the canonical model for PDL$^\Box$) but instead when $\sigma(W_1)\land\langle\program{p}\rangle \sigma(W_2)\not\vdash\bot$. 
	
	One should think about sections as a way of navigating this dilemma in the following way: when working with a section, you first `forget' everything dynamic about your states, and then you replace this information with the dynamic information delivered by the section. More helpfully, in the construction below, what will happen is that the modal behavior of each point $W$ in our canonical models will be determined not by the modal information $W$ actually does contain but instead by the modal information in $\sigma(W)$ where $\sigma$ is an appropriate section.
	
	By keeping the members of $S(X)$ as the points in our model, we ensure that we have sufficiently many points (and, in particular, sufficiently many locally indistinguishable points) on hand. But by tying the dynamic behavior of our points to sections rather than to their actual contents, we ensure that the `anything goes' flavor of PEDAL's approach to program valuations is preserved. 
	
	Note that if $\sigma$ is a section then the restriction, $\sigma\vert_W$, of $\sigma$ to $\dmff(W)$ is a function $\dmff(W)\to\dmff(W)$. We say that a section is \emph{canonical} when for all $W$, the function $\sigma\vert_W$ is a bijection. Let $\can(X)$ be the set of all canonical sections of $S(X)$ and let $\delta(W)$ be the conjunction of all dmffs in $W$.
	
	\begin{lemma}
		$\langle W,V\rangle\in\llbracket\pi\rrbracket^\canon$ iff $\langle\sigma^{-1}(W),\sigma^{-1}(V)\rangle\in\llbracket\pi\rrbracket_\sigma$
	\end{lemma}
	\begin{proof}
		$\llbracket -\rrbracket^\canon$ and $\llbracket -\rrbracket_\sigma$ are defined by the exact same recursive extension of what they do to atomic programs. So it suffices to ensure that for all atomic programs $\program{p}$, $\langle W,V\rangle\in\llbracket\program{p}\rrbracket^\canon$ iff $\langle\sigma^{-1}(W),\sigma^{-1}(V)\rangle\in\llbracket\program{p}\rrbracket_\sigma$. But this is easily established since $\langle\sigma^{-1}(W),\sigma^{-1}(V)\rangle\in\llbracket\program{p}\rrbracket_\sigma$ iff $\sigma(\sigma^{-1}(W))\land\langle\program{p}\rangle\sigma(\sigma^{-1}(V))\not\vdash\bot$ iff $W\land\langle\program{p}\rangle V\not\vdash\bot$ iff $\langle W,V\rangle\in\llbracket\program{p}\rrbracket^\canon$.
	\end{proof}
	
	\begin{lemma}\label{canonical}
		Let $X$ be a finite, consistent set of $\LL_1$-formulas and let $X'=\{\gamma:\Pr(\gamma)\geq q$ occurs in $A$ for some $A\in X\}$. Let $M^\canon(X')$ be the canonical PDL$^\Box$-model for $X'$. For each canonical section $\sigma$, let $M(\sigma)$ be the model $\langle S(X),v^\canon_\form,v_\sigma,R^\canon_\Box\rangle$. Then $M^\canon(X'),W\vDash \gamma$ iff $M(\sigma),\sigma^{-1}(W)\vDash \gamma$. 
	\end{lemma}
	\begin{proof}
		By induction on $\gamma$. Since $\sigma$ is canonical, $\sigma^{-1}(W)$ is dmff-compatible with $W$. It follows that $W\in v^\canon_\form(p)$ iff $\sigma^{-1}(W)\in v^\canon_\form(p)$, from which the base case follows immediately. For the induction, only the $\langle\pi\rangle$-case is worth explicitly examining. 
		
		Note that $M^\canon(X'),W\vDash\langle\pi\rangle\gamma$ iff for some $V\in S(X')$, $\langle W,V\rangle\in\llbracket\pi\rrbracket^\canon$ and $M^\canon(X'),V\vDash\gamma$. But then by the previous lemma and the inductive hypothesis, $\langle \sigma^{-1}(W),\sigma^{-1}(V)\rangle\in\llbracket\pi\rrbracket_\sigma$ and $M(\sigma),\sigma^{-1}(V)\vDash\gamma$. So $M(\sigma),\sigma^{-1}(W)\vDash\langle\pi\rangle\gamma$.
	\end{proof}
	
	\begin{lemma}\label{cut_down_lemma}
		For all $\gamma\in\FL^{\pm}(X)$, $\vdash(\delta(W)\land\gamma)\otto\bigvee\{\sigma(W):\sigma\in\can(X)\text{ and }\gamma\in\sigma(W)\}$.
	\end{lemma}
	\begin{proof}
		By construction, $\{\sigma(W):\sigma\in\can(X)\}$ is exactly $\dmff(W)$. The result then follows from Lemma~\ref{equiv_lemma} essentially immediately. 
	\end{proof}
        
   \begin{lemma}\label{it_all_works}
   		Let $\sigma$ be a canonical section of $S(X)$ and write $M^\canon_\PEDAL(X)(v_\sigma)$ for the tuple $\langle S(X), v^\canon_\form,v_\sigma,R^\canon_\Box\rangle$.
   		\begin{enumerate}
			\item If $\neg\gamma\in\FL^{\pm}(X)$ then $\gamma\not\in\sigma(W)$ iff $\neg\gamma\in \sigma(W)$;
			\item If $\gamma_1\lor\gamma_2\in\FL^{\pm}(X)$ then $\gamma_1\lor\gamma_2\in\sigma(W)$ iff $\gamma_1\in\sigma(W)$ or $\gamma_2\in\sigma(W)$. 
			\item If $W\in S(X)$ and $W'\in S(X)$, then for all programs $\pi$, if $\sigma(W)\land\langle\pi\rangle\sigma(W')\not\vdash\bot$, then $\langle W,W'\rangle\in \llbracket\pi\rrbracket_\sigma$.
	    	\item For all $W\in S(X)$, if $\langle\program\pi\rangle\gamma\in\FL^{\pm}(X)$ then $\langle\pi\rangle\gamma\in \sigma(W)$ iff for some $W'\in S(X)$, $\langle W,W'\rangle\in \llbracket\pi\rrbracket_\sigma$ and $\gamma\in \sigma(W')$.
			\item For all $W\in S(X)$, $\Box\gamma\in\sigma(W)$ iff for all $W'\in S(X)$, if $\langle W,W'\rangle\in R^\canon_\Box$ then $\gamma\in\sigma(W')$.
			\item $M^\canon_\PEDAL(X)(v_\sigma)$ is a PDL$^\Box$-model.
		\end{enumerate}
	\end{lemma}
	\begin{proof}\mbox{}
		Using Lemma~\ref{canonical}, these are straightforward consequences of Lemmas~\ref{connectiveswork}--\ref{pdl_box_key_lemma}. 
	\end{proof}	
	
	We also state the following results. Their proofs are entirely standard; see e.g. Theorems 3.1, 3.3, and 3.4 and Lemma 3.3 of \cite{ognjanovic2016probability}:\smallskip
	\begin{description}\setlength
		\item[Deduction Theorem] If $X,A\vdash B$, then $X\vdash A\to B$. 
		\item[Lindenbaum Lemma] If $X$ is a consistent set of $\LL_1$-formulas, then there is a maximal consistent set of $\LL_1$-formulas $X'\supseteq X$. 
		\item[Probability Rules] For all $s\geq r$, 
		\begin{itemize}[leftmargin=*]
			\item $\vdash \Pr(\gamma_1\to\gamma_2)=1\to(\Pr(\gamma_1)\geq s\to\Pr(\gamma_2)\geq s)$.
			\item If $\vdash\Pr(\gamma_1\otto\gamma_2)=1$, then $\vdash\Pr(\gamma_1)\geq s\otto\Pr(\gamma_2)\geq s$.
			\item $\vdash\Pr(\gamma_1)\geq s\to\Pr(\gamma_1)\geq r$.
			\item $\vdash\Pr(\gamma_1)\leq r\to\Pr(\gamma_1)\leq s$.
			\item If $T$ is a maximal consistent set of formulas and $Q=\sup\{q:\Pr(\gamma)\geq q\in T\}$ and $Q\in\mathbb{Q}$, then $\Pr(\gamma)\geq Q\in T$.
		\end{itemize}
	\end{description}\smallskip
	
	\begin{lemma}
		If $\{\gamma_i\}_{i=1}^n$ is a set of PDL$^\Box$-formulas and $\gamma_i\land\gamma_j\vdash_{\text{PDL}^\Box}\bot$ for $i\neq j$, then $ \vdash\bigwedge_{i=1}^n(\Pr(\gamma_i)\geq r_i)\to\Pr\left(\bigvee_{i=1}^n\gamma_i\right)\geq\sum_{i=1}^n r_i$	
	\end{lemma}
	\begin{proof}
		By induction on $n$. 
		
%		The base case is immediate and the $n=2$ case follows from axiom A6 given that $\vdash\Pr(\gamma_1\land\gamma_2)=0$ follows from $\gamma_1\land\gamma_2\vdash_{\text{PDL}^\Box}\bot$. 
%		
%		Suppose $\vdash\bigwedge_{i=1}^k(\Pr(\gamma_i)\geq r_i)\to\Pr\left(\bigvee_{i=1}^k\gamma_i\right)\geq\sum_{i=1}^k r_i$. Then $\vdash\bigwedge_{i=1}^{k+1}(\Pr(\gamma_i)\geq r_i)\to\left(\Pr\left(\bigvee_{i=1}^k\gamma_i\right)\geq\sum_{i=1}^k r_i\land\Pr(\gamma_{k+1})\geq r_{k+1}\right)$. But since $\gamma_i\land\gamma_{k+1}\vdash\bot$ for $i\leq k$ we also have that $\left(\bigvee_{i=1}^k\gamma_i\right)\land\gamma_{k+1}\vdash\bot$. Thus $\vdash\Pr\left(\left(\bigvee_{i=1}^k\gamma_i\right)\land\gamma_{k+1}\right)=0$
%		
%		Altogether this gives us that
%		\begin{displaymath}
%			\vdash\bigwedge_{i=1}^{k+1}(\Pr(\gamma_i)\geq r_i)\to\left(\Pr\left(\bigvee_{i=1}^k\gamma_i\right)\geq\sum_{i=1}^k r_i\land\Pr(\gamma_{k+1})\geq r_{k+1}\land \Pr\left(\left(\bigvee_{i=1}^k\gamma_i\right)\land\gamma_{k+1}\right)=0\right).
%		\end{displaymath}
%		But the righthand side here is the antecedent of an instance of A6. So by propositional logic we get $\vdash\bigwedge_{i=1}^{k+1}(\Pr(\gamma_i)\geq r_i)\to\Pr\left(\bigvee_{i=1}^{k+1}\gamma_i\right)\geq\sum_{i=1}^{k+1} r_i$, which completes the induction. 
	\end{proof}
			
	\begin{lemma}\label{PEDAL_key}
		If $\gamma\in\FL^{\pm}(X)$, then $M^\canon_\PEDAL(X)(v_\sigma),W\vDash\gamma$ iff $\gamma\in\sigma(W)$. 
	\end{lemma}
	\begin{proof}
		By induction on $\gamma$. All cases except the $\langle\pi\rangle$-case can be dealt with just as they were in Lemma~\ref{pdl_box_key_lemma}. For the $\langle\pi\rangle$-case, observe that $M^\canon_\PEDAL(v_\sigma),W\vDash\langle\pi\rangle\gamma$ iff for some $V\in S(X)$ we have both $\langle W,V\rangle\in\llbracket\pi\rrbracket_\sigma$ and $M^\canon_\PEDAL(v_\sigma),V\vDash\gamma$. But this happens iff (by the inductive hypothesis) $\gamma\in\sigma(V)$ which happens iff (by part 4 of Lemma~\ref{it_all_works}) $\langle\pi\rangle\gamma\in\sigma(W)$. 
	\end{proof}
	
	Let $X$ be a consistent set of $\LL_1$-formulas, $X'$ a maximal consistent extension of $X$, and $X\pda$ be the set of dmffs $\gamma$ in $\FL^{\pm}(X)$ such that $X'\vdash\Pr(\gamma)=1$. Note that since $X$ is consistent, axiom A3 guarantees that there is $W\in S(X)$ with exactly the same dmffs as $X\pda$. Choose one such and abuse notation to call it $X\pda$ too. Finally, where $W\in S(X)$, let $P_W=\{v_\sigma:\sigma\in\can(X)\text{ and }\sigma(X\pda)=W\}$.
	
	The canonical model $M^\canon_\PEDAL(X)$ is the tuple $\langle S(X),v^\canon_\form,R^\canon_\Box,\Pi^\canon_X,\mu^\canon_X\rangle$ where
	\begin{itemize}\setlength
		\item $v^\canon_{\form}(p)=\{W\in S(X) : p\in W\}$
		\item $R^\canon_\Box=\{\langle W_1,W_2\rangle: W_1\land\Diamond W_2\not\vdash\bot\}$
		\item $\Pi^\canon_X=\{P_W : W\in S(X)\}\cup\{\{\text{everything else}\}\}$.
		\item $\mu^\canon_X(P_W)=\sup\{q : X'\vdash\Pr(W)\geq q\}$.
	\end{itemize}
	The first three elements of this model are what one should expect them to be. The partition, $\Pi^\canon_X$, recognizes only the different ways the dynamic behavior of $X\pda$ might be determined and ignores everything else that might possibly vary in the canonical model. And, for each of these possibilities, the measure $\mu^\canon_X$ assigns it what is essentially exactly the probability that $X'$ says it should. What remains is to check that (a) this actually settles things in a way that makes sense and that (b) the result really does model $X$.
		
	\begin{lemma}
		If $X$ is finite and consistent, then $M^\canon_\PEDAL(X)$ is a PEDAL-model.
	\end{lemma}
	\begin{proof}
		Most pieces are obvious. Of those that aren't, notice that since the set of canonical sections is finite, it's clear that $\Pi$ and $\mu$ meet the finiteness assumptions. To see that $\mu$ is a probability function, first let $\dmff(X\pda)=\{W_1,\dots,W_n\}$ and note that $\sum_{P\in\Pi^\canon_X}\mu(P)=\sum_{i=1}^n\sup\{q: X\vdash\Pr(W_i)\geq q\}$. Also recall from Lemma~\ref{cut_down_lemma} that $\vdash\delta(X\pda)\otto\bigvee_{W\in\dmff(X\pda)} W$. By definition, $X\vdash\Pr(\delta(X\pda))=1$. So $X\vdash\Pr(\bigvee_{i=1}^{n} W_i)=1$. 
		
		Now suppose that $\sum_{i=1}^n\sup\{q: X\vdash\Pr(W_i)\geq q\}<1$. Then there are $Q_1,\dots,Q_n\in[0,1]\cap\mathbb{Q}$ so that $Q_i\geq\sup\{q:X\vdash\Pr(W_i)\geq q\}$ and $\sum_{i=1}^n Q_i<1$. Since $Q_i\geq\sup\{q:X\vdash\Pr(W_i)\geq q\}$ and $X$ is maximal, $X\vdash\Pr(W_i)\leq Q_i$. So by repeated application of A7, we get that $X\vdash\Pr(\bigvee_{i=1}^n W_i)\leq\sum_{i=1}^n Q_i$. But this impossible because $X$ is consistent, $\sum_{i=1}^n Q_i<1$, and $X\vdash\Pr(\bigvee_{i=1}^{n} W_i)=1$. So $\sum_{i=1}^n\sup\{q: X\vdash\Pr(W_i)\geq q\}=1$.
	\end{proof}
		
	\begin{lemma}
		For $A\in \FL^{\pm}(X)$, $M^\canon_\PEDAL,X\pda\vDash A$ iff $X\vdash A$. 
	\end{lemma}
	\begin{proof}
		By induction on $A$. Most of the action is in the base case which is the only case we'll examine. 
		
		If $v\in P_W\cap \Val^{S(X)}_\prog(X\pda,\gamma)$ then for some section $\sigma$ with $\sigma(X\pda)=W$, $v=v_\sigma$ and $M^\canon_\PEDAL(X)(v_\sigma),X\pda\vDash\gamma$. But by Lemma~\ref{PEDAL_key}, this happens iff $\gamma\in W=\sigma(X\pda)$. So $P_W\cap \Val^{S(X)}_\prog(X\pda,\gamma)$ is either empty (if $\gamma\not\in W$) or all of $P_W$ (if $\gamma\in W$). By definition, $M^\canon_\PEDAL(X),X\pda\vDash\Pr(\gamma)\geq q$ iff $\mu(\Val^{S(X)}_\prog(X\pda,\gamma))\geq q$. So
		\begin{align*}
			\mu(\Val^{S(X)}_\prog(X\pda,\gamma)) 
			& =\sum_{W\in S(X)}\mu(P_W)\frac{|P_W\cap \Val^{S(X)}_\prog(X\pda,\gamma)|}{|P_W|} \\
			& =\sum\{\mu(P_W):\gamma\in W\in\dmff(X\pda)\} \\
			& = \sum\{\sup\{q: X\vdash\Pr(W)\geq q\}:\gamma\in W\in\dmff(X\pda)\}
		\end{align*}  
		From here the argument proceeds like the argument in the previous lemma. 
	\end{proof}
	
	\begin{theorem}
		PEDAL is complete for PEDAL-models.
	\end{theorem}
	\begin{proof}
		Suppose $A$ isn't provable. Then $\neg A$ is consistent. So $\{\neg A\}$ can be extended to a maximal consistent set, from which we can construct a model where $\neg A$ is satisfied by the above results. 
	\end{proof}
	
	\section{Addressing the Infinitary Rule}
	
	I think PEDAL is worth taking seriously in spite of its infinitary nature, both because what it offers (a way to help us to reason well about a situation we ubiquitously find ourselves in where good reasoning is difficult) and because the infinitary rule itself is one we often can, in fact, apply. That said, it would clearly be better to have a finitary proof system.
	
	One way we might do this is to adopt the techniques examined in, e.g. \cite{fattorosi1987modal} among other places and restrict the range of $\mu$ to a finite subset of the rationals between zero and one. The resulting systems can, in the case of other probabilistic logics (see \cite[\S5.2]{ognjanovic2016probability} for another example), be finitely axiomatized. So it seems quite likely that this will also be the case for our system. %In general, these axiomatizations have a single implausible axiom that says, in essence that if $\Pr(\gamma)>\alpha$ where $\alpha$ is one of the values the semantics permits probabilities to take, then $\Pr(\gamma)\geq\alpha^+$, where $\alpha^+$ is `the next' probability that is permitted. 
	
%	This is a problem that is, in practice, either tolerable or easily avoided. It's tolerable to the extent that one has fine-tuned the range of permissible probabilities in such a way that the differences between them really are the sort of thing one can usually ignore. It's avoidable to the extent that one can detect when one isn't being sufficiently sensitive and move to a larger---though still finite---set of probabilities.
	
	The second option is to bound the number of states allowed in our models. In \cite{abadi1994decidability}, it was shown that in the non-dynamic but first-order case this also led to a finitely axiomatized system. %As above, there are axiomatic downsides to this approach, but also as above they can usually be either tolerated or avoided. Should either this approach or the previous approach fail to achieve desired results, there's not a clear reason to avoid seeing what happens if one combines them.
	
	The most speculative solution is to adopt an analogue of the `schematic variable' approaches used in e.g. \cite{read2004identity} and \cite{read2016harmonic}. The loose idea is that one allows assumptions of the form `$x_n<q_n$' to occur as assumptions in proofs, to be discharged when the proof reaches a formula of the form $\Pr(\gamma)\leq x_n\lor B$ with $x_n\not\in B$. After discharge, the next line of the proof is $\Pr(\gamma)\leq q_n$. Stated more briefly but also yet more loosely, one allows the effect of one sort of restricted quantification, without allowing the full power (and decision-theoretic problems) that quantifiers permit.
	
	Each of these is an avenue for future research that, with the system in this paper on the table, is possible to approach. Each seems like it may quite plausibly result in an implementable version of PEDAL. Unfortunately, they will all have to wait for future work.
		
	\section{Conclusion}
	
	Recall the question we began with: supposing Anne is 60\% certain that $\Box(A\to[\program{p}]B)$ is true and 60\% certain that $\Box(C\to[\program{q}]D)$ is true, how confident should she be that $\Box[\program{(?A;p)\cup(?C;q)}](B\lor D)$ is true? We can now give a rigorous answer to this question. 
	
	First let's say that a confidence lower bound (clb) is any $q$ for which PEDAL will derive $\Pr(\Box[\program{(?A;p)\cup(?C;q)}](B\lor D))\geq q$ from $(\Pr(\Box(A\to[\program{p}]B))\geq 0.6$ and $\Pr(\Box(C\to[\program{q}]D))\geq0.6)$. Given that PEDAL-models accurately model Anne's epistemic state, she should be at least $q$-confident that $\Box[\program{(?A;p)\cup(?C;q)}](B\lor D)$ is true for any clb $q$. And she should be \emph{no more than} $Q$ confident that $\Box[\program{(?A;p)\cup(?C;q)}](B\lor D)$ where $Q$ is the supremum of all the clbs $q$.
	
	Next note that (while the derivation is nontrivial) it is in fact the case that if $\vdash_{\text{PDL}^\Box}A\to B$, then $\vdash\Pr(A)\geq q\to\Pr(B)\geq q$ for all appropriate $q$. The reader attempting to reconstruct the precise derivations we now gesture at will need to use this fact often.
	
	To begin, notice that $\vdash \Pr(\gamma)\geq 0.6\to\Pr(\neg\gamma)\leq 0.4$. Using this fact twice, Axiom A7 and some elbow grease gives that 0.2 is a clb for Anne. 
	
	To see that 0.2 is the supremum of Anne's clbs, we give a model where both $(\Pr(\Box(A\to[\program{p}]B))\geq 0.6$ and $\Pr(\Box(C\to[\program{q}]D))\geq0.6)$ are true and so is $\Pr(\Box[\program{(?A;p)\cup(?C;q)}](B\lor D))=0.2$. By soundness, it follows that no higher bound can be proved. To make the display of such a model simpler, we suppose that $A$, $B$, $C$, and $D$ are atomic. We use a model with 2 states, $s$ and $t$ and use universal accessibility relation for $R_\Box$. Let $A$ and $C$ be true in $s$ but not $t$ and $B$ and $D$ be true in $t$ but not $s$. Suppose $v_1$, $v_2$, and $v_3$ are the following valuations:
	\begin{displaymath}
		v_1:=
		\left(
			\xymatrix@=3mm{
			s\ar[r]^{\program{p}}\ar@(ul,dl)[]_{\program{q}} & t \\
			}
		\right)
		\qquad
		v_2:=
		\left(
			\xymatrix@=3mm{
			s\ar@/^/[r]^{\program{p}}\ar@/_/_{\program{q}}[r] & t \\
			}
		\right)
		\qquad
		v_3:=
		\left(
			\xymatrix@=3mm{
			s\ar[r]^{\program{q}}\ar@(ul,dl)[]_{\program{p}} & t \\
			}
		\right)
	\end{displaymath}
	Notice that in $v_1$ and in $v_2$, both states make $(A\to[\program{p}]B)$ true. Also notice that in $v_2$ and $v_3$ both states make $(C\to[\program{q}]D)$ true. But in $v_1$ we can execute $\program{q}$ at the $C$-state $s$, and doing so lands us back in $s$, which is not a $D$-state. Similarly, in $v_3$ we can execute $\program{p}$ at the $A$-state $s$, and doing so lands us back in $s$, which is not a $B$-state. It follows that in $v_1$ and in $v_2$ (but not in $v_3$), $\Box(A\to[\program{p}]B)$ is true and that in $v_2$ and in $v_3$ (but not in $v_1$), $\Box(C\to[\program{q}]D)$ is true. 
	
	Finally, observe that in $v_1$ and $v_3$, there is an execution of $\program{(?A;p)\cup(?C;q)}$ at $s$ that returns to the non-$B\lor D$-state $s$. In $v_1$, it's the execution that tests whether $C$, then executes $\program{q}$; in $v_3$ it's the execution that tests whether $A$, then executes $\program{p}$. So it's \emph{only} in $v_2$ that $[\program{(?A;p)\cup(?C;q)}](B\lor D)$ is true at both states.
	
	Now let $\Pi$ be the partition $\{\{v_1\},\{v_2\},\{v_3\},\{\text{everything else}\}\}$ and let $\mu$ assign 40\% probability to the cells $\{v_1\}$ and $\{v_2\}$, 20\% probability to the cell $\{v_3\}$ and 0 probability to the remaining cell. On the resulting PEDAL-model, both states make $\Pr(\Box(A\to[\program{p}]B))=0.6$  true, and make $\Pr(\Box(C\to[\program{q}]D))=0.6$ true, and also make $\Pr(\Box[\program{(?A;p)\cup(?C;q)}](B\lor D))=0.2$ true. 

	It follows that 0.2 is the supremum of Anne's clbs. As a result, should Anne think that her credences concerning the behavior of $\program{p}$ and $\program{q}$ ought to warrant her greater confidence than this, she owes an explanation of what, precisely, justifies this additional credence since no credence greater than 0.2 is warranted by the judgments she has expressed. 

\begin{credits}
\subsubsection{\ackname} Conversations with Will Stafford and Bruce Glymour greatly improved this paper. 

\subsubsection{\discintname}
No competing interests to declare
\end{credits}
%
% ---- Bibliography ----
%

\bibliographystyle{splncs04}
\bibliography{biblio}

\end{document}